\documentclass[sigconf,screen]{acmart}

\AtBeginDocument{
  }

\setcopyright{cc}
\setcctype{by}
\acmDOI{10.1145/3832783.3834624}
\acmYear{2026}
\copyrightyear{2026}
\acmISBN{979-8-4007-2882-2/2026/10}
\acmConference[ASE '26]{Proceedings of the 41st IEEE/ACM International Conference on Automated Software Engineering}{October 12--16, 2026}{Munich, Germany}
\acmBooktitle{Proceedings of the 41st IEEE/ACM International Conference on Automated Software Engineering (ASE '26), October 12--16, 2026, Munich, Germany}
\acmSubmissionID{ase26tool-p77-p}
\received{2026-05-11}
\received[accepted]{2026-06-19}

\input{marco}

\begin{document}

\title{MIMOS-Tools: An Integrated Toolchain for
Model-Based Design of Embedded Systems}

\author{Wang Yi}
\correspondingauthor
\orcid{0000-0002-2994-6110}
\affiliation{%
  \institution{Uppsala University}
  \department{Department of Information Technology}
  \city{Uppsala}
  \country{Sweden}
}
\email{wang.yi@it.uu.se}

\author{Chengzi Huang}
\orcid{0009-0001-6326-2636}
\affiliation{%
  \institution{Uppsala University}
  \department{Department of Information Technology}
  \city{Uppsala}
  \country{Sweden}
}
\email{chengzi.huang@it.uu.se}

\author{Behnam Khodabandeloo}
\orcid{0009-0004-6118-4429}
\affiliation{%
  \institution{Uppsala University}
  \department{Department of Information Technology}
  \city{Uppsala}
  \country{Sweden}
}
\email{behnam.khodabandeloo@it.uu.se}

\author{Duc Anh Nguyen}
\orcid{0009-0006-0899-7513}
\affiliation{%
  \institution{Uppsala University}
  \department{Department of Information Technology}
  \city{Uppsala}
  \country{Sweden}
}
\email{ducanh.nguyen@it.uu.se}

\renewcommand{\shortauthors}{Yi, Huang, Khodabandeloo, and Nguyen}

\begin{abstract}
Despite decades of progress in model-based methods and tools for embedded real-time systems, there is still no integrated environment that supports the entire development lifecycle. This paper presents MIMOS-Tools, a comprehensive toolchain spanning the key phases of development, including modeling, simulation, formal verification, real-time scheduling, schedulability and end-to-end latency analysis, and code generation for heterogeneous multicore platforms. 
All phases are grounded in a single, coherent, and deterministic model of computation, the MIMOS model, for which we present a formal operational semantics and establish both functional and timing determinism.
This semantic foundation ensures that properties verified at the model level are preserved throughout the development process and in the generated implementation.

A distinguishing feature of MIMOS-Tools is its support for not only deterministic but also modular, compositional design. Unlike synchronous design environments such as SCADE and Simulink, which rely on a global base clock and offline static scheduling, MIMOS-Tools supports time-triggered asynchronous communication together with flexible real-time scheduling of computation. This removes the tight global coupling of the synchronous paradigm: components can be added, modified, or updated without recomputing a system-wide schedule, enabling scalable and incremental design and dynamic updates throughout the system lifecycle. However, MIMOS-Tools also supports Simulink models when the deadlines of tasks are set to zero to emulate the zero-time semantics of synchronous design.

The toolchain is available at \url{https://uu-mimos.github.io}. 
\end{abstract}



\maketitle

\section{Introduction}
\label{sec:Introduction}

Over the past decades, many theories, methods, and tools have been developed for the design, analysis, and deployment of embedded real-time systems. 
Major advances include modeling and simulation environments \cite{mathworks2026simulink, fritzson1998modelica, ansys2026scade}, formal verification techniques \cite{larsen1997uppaal}, 
real-time scheduling and analysis \cite{baruah2015multiprocessor},
time-triggered architectures \cite{kopetz2003time,henzinger2003giotto}, 
and synchronous programming languages and industrial design tools \cite{halbwachs1991lustre, ansys2026scade, mathworks2026simulink}.
Yet no single environment covers the full development lifecycle and later system evolution.
Existing solutions usually address
only specific development stages and focus on particular
classes of properties. For instance, formal verification often
relies on abstractions that omit implementation details,
creating a semantic gap between verified models and deployed
systems. Conversely, real-time scheduling theory concentrates
on non-functional aspects such as workload characterization and
resource allocation while largely abstracting from functional
behavior. As a consequence, engineers have to manually integrate
heterogeneous models and tools across multiple abstraction
levels, a process that is error-prone and difficult to certify
in safety-critical contexts.

This paper presents MIMOS-Tools, an integrated toolchain that closes this  gap by grounding the entire development workflow in a single deterministic model of computation: the MIMOS model~\cite{yi2022mimos}, inspired by Kahn  Process Networks~\cite{kahn1974semantics} and augmented with explicit timing constraints in the spirit of Logical Execution Time~\cite{henzinger2003giotto}. 
For this model, we present a formal operational semantics and a determinism theorem, which serve as the single semantic reference for every stage of the toolchain.
Unlike synchronous design tools such as SCADE and Simulink~\cite{ansys2026scade, mathworks2026simulink}, which rely on a global base clock and offline static schedules, MIMOS-Tools eliminates the tight global coupling of the synchronous
paradigm: adding a task with a new period does not force
recomputation of a system-wide schedule, enabling truly modular
and incremental design on modern multicore and heterogeneous
platforms.

The toolchain covers four key capabilities.
\textit{Graphical modeling} provides a visual, hierarchical environment with mixed-language support and LLM-assisted programming for rapid prototyping and iterative design.
\textit{Simulation and verification} offers fast functional analysis and validation through simulation, complemented by formal verification to guarantee correctness properties.
\textit{Real-time scheduling and analysis} maps design models to real-time tasks and supports multicore scheduling, schedulability and end-to-end latency analysis.
\textit{Target code generation} performs a semantics-preserv\-ing translation from MIMOS models to executable multicore code whose execution retains the functional and timing properties verified in the earlier stages.

\section{The MIMOS Model}
\label{sec:model_architecture}

MIMOS~\cite{yi2022mimos} is a formalism for modeling, implementing, and updating embedded real-time systems, conceptually a \emph{multi-rate task graph} model.
In this section we present its operational semantics, which serves as the formal foundation shared by all features of MIMOS-Tools.
We begin with an informal overview, followed by the transition rules and a determinism theorem guaranteeing that systems specified in MIMOS behave deterministically in both functionality and timing.

\subsection{Informal Semantics}
\label{sec:model}
In MIMOS, a system is modeled as a directed graph where nodes represent computation \emph{tasks} and edges represent communication \emph{channels}. Tasks communicate exclusively via channels by exchanging \emph{tokens} (values of specific data types).

Each task may have multiple input and output channels, and has a \emph{task function} that defines a mapping from input tokens to output tokens. 
Tasks are periodic\footnote{The tasks can be described by any well-established real-time task model with deterministic release patterns.}, with a \textit{period} between consecutive releases and a \textit{relative deadline}.
A worst-case execution time (WCET) is also given for scheduling analysis and code generation; it does not affect the model's timing behavior, which depends only on periods and deadlines under the assumption that tasks are schedulable.

Channels are \textit{single-writer and single-reader}.
One task writes to each channel and one task reads from it\footnote{A channel may be copied for multiple readers without changing the model.}.
Each channel buffers tokens and is either a \textit{FIFO} or a \textit{register}.
A FIFO channel is an unbounded order-preserving buffer initialized with any finite number of tokens.
A write appends one token to the tail, and a read removes and returns one from the head.
A register channel holds one token and is always initialized.
A write overwrites the token, and a read returns it without removing it.

Every activation follows a \emph{read--compute--write} pattern.
At its release instant, a task is \emph{enabled} if and only if each FIFO input holds at least one token (registers are always available).
Otherwise the activation is \emph{skipped}.
In this case no tokens are consumed, no outputs are produced, and the task's inputs are not checked again until the next release instant.
When enabled, the task consumes its inputs at release, executes at any point chosen by the scheduler before the deadline, and at the deadline produces its outputs.
Both reads and writes are logically instantaneous.
They take zero time and occur at the release and deadline instants.
When reads and writes occur at the same instant (possibly by different tasks), all writes are performed before any reads (the \textit{write-before-read} rule).
This abstraction ensures that reads and writes happen at predefined time instants given by the periods and deadlines, so their order does not depend on how the computations are scheduled.

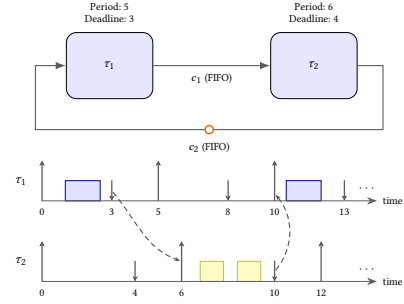
\begin{figure}[!htbp]
  \centering
  \begin{subfigure}[b]{0.46\textwidth}
    \centering
  \begingroup
  \resizebox{0.57\linewidth}{!}{
\begin{tikzpicture}[>=Latex, font=\footnotesize]
  \tikzset{
    task/.style={draw, rounded corners=2mm, fill=blue!10,
                 minimum width=1.7cm, minimum height=1.3cm, line width=0.5pt},
    token/.style={draw=orange!85!black, fill=white, circle, thick, inner sep=1.6pt}
  }

  \node[task] (t1) {$\tau_1$};
  \node[task, right=2.3cm of t1] (t2) {$\tau_2$};

  \node[above=0.8mm of t1, align=center, font=\scriptsize]
       {Period: 5\\Deadline: 3};
  \node[above=0.8mm of t2, align=center, font=\scriptsize]
       {Period: 6\\Deadline: 4};

  \draw[->, semithick, draw=gray!70!black] (t1.east) -- (t2.west);
  \node[above=0.3mm, font=\scriptsize] at ($(t1.east)!0.10!(t2.west)$) {};
  \node[above=0.3mm, font=\scriptsize] at ($(t1.east)!0.90!(t2.west)$) {};
  \coordinate (mid) at ($(t1.east)!0.5!(t2.west)$);
  \node[below=0.4mm, font=\scriptsize] at (mid) {$c_1$ (FIFO)};

  \coordinate (L) at ($(t1.west)+(-6mm,0)$);
  \coordinate (R) at ($(t2.east)+(5mm,0)$);
  \coordinate (B) at ($(t1.south)!0.5!(t2.south) + (0,-6mm)$);

  \draw[semithick, draw=gray!70!black, ->]
        (t2.east) -- (R) -- (R |- B) -- (L |- B) -- (L) -- (t1.west);

  \node[font=\scriptsize] at ($(L)+(4mm,2.2mm)$) {};
  \node[font=\scriptsize] at ($(R)+(-4mm,2.2mm)$) {};

  \node[token] at ($(L|-B)!0.5!(R|-B)$) (tk2) {};
  \node[font=\scriptsize, below=0.4mm of tk2] {$c_2$ (FIFO)};
\end{tikzpicture}}%
  \endgroup

  \end{subfigure}
  \hfill
  \begin{subfigure}[b]{0.46\textwidth}
    \centering
  \begingroup
  \resizebox{0.65\linewidth}{!}{
\begin{tikzpicture}[
  x=0.45cm, y=0.78cm, >=Stealth, font=\footnotesize,
  axis/.style={->, line width=0.55pt, gray!75!black},
  rel/.style={gray!70!black, line width=0.7pt, line cap=round,
              -{Stealth[length=3pt,width=3pt]}},
  ddl/.style={gray!70!black, line width=0.7pt, line cap=round,
              -{Stealth[length=3pt,width=3pt]}},
  job1/.style={fill=blue!12, draw=blue!55!black, line width=0.5pt},
  job2/.style={fill=yellow!28, draw=yellow!72!black, line width=0.5pt},
  transfer/.style={densely dashed, line width=0.65pt, gray!70!black,
                   -{Stealth[length=3pt]}, shorten >=1pt, shorten <=1pt}
]

\def\yA{0}        \def\yB{-2.0}
\def\h{0.50}                 
\def\arrTop{1.00}            
\def\dTail{0.5}              

\node[font=\footnotesize] at (-0.9, \yA+0.50) {$\tau_1$};
\node[font=\footnotesize] at (-0.9, \yB+0.50) {$\tau_2$};

\foreach \a/\b in {1/2.5, 10.5/12}{ \draw[job1] (\a,\yA) rectangle (\b,\yA+\h); }
\foreach \a/\b in {6.8/7.8, 8.4/9.4}{ \draw[job2] (\a,\yB) rectangle (\b,\yB+\h); }

\draw[axis] (0,\yA) -- (14.4,\yA) node[right, black, font=\scriptsize] {time};
\draw[axis] (0,\yB) -- (14.4,\yB) node[right, black, font=\scriptsize] {time};

\foreach \t in {0,5,10}{ \draw[rel] (\t,\yA) -- (\t,\yA+\arrTop); }
\foreach \t in {3,8,13}{ \draw[ddl] (\t,\yA+\dTail) -- (\t,\yA); }
\foreach \t in {0,3,5,8,10,13}{ \node[below=1pt, font=\scriptsize] at (\t,\yA) {\t}; }

\foreach \t in {0,6,12}{ \draw[rel] (\t,\yB) -- (\t,\yB+\arrTop); }
\foreach \t in {4, 10}{ \draw[ddl] (\t,\yB+\dTail) -- (\t,\yB); }
\foreach \t in {0,4,6,10,12}{ \node[below=1pt, font=\scriptsize] at (\t,\yB) {\t}; }

\draw[transfer] (3,\yA+\h/2) to[out=-50,in=160] (6,\yB+\h);
\draw[transfer] (10,\yB+\h/2) to[out=50,in=-50] (10,\yA+0.18);

\node[font=\scriptsize] at (14,\yA+0.28) {$\cdots$};
\node[font=\scriptsize] at (14,\yB+0.28) {$\cdots$};
\end{tikzpicture}}%
  \endgroup

  \end{subfigure}
  \caption{A two-task MIMOS model with FIFOs $c_1$ (initially empty) and $c_2$ (one initial token), and its execution timeline. Releases: up-arrows; deadlines: shorter down-arrows; computations: colored blocks; communication: dashed arrows.
}
  \Description{Top: a two-task MIMOS model. Bottom: a time-line showing releases, deadlines, computation blocks, and data transfers between the tasks.}
  \label{fig:mimos_example}
\end{figure}

\myparagraph{Example} Figure~\ref{fig:mimos_example} shows tasks $\tau_1$ and $\tau_2$ connected by FIFOs $c_1$ (initially empty) and $c_2$ (one initial token).
At time~0 both tasks are released.
$\tau_1$ reads the token of $c_2$ and writes $c_1$ at its deadline, time~3.
$\tau_2$ skips since $c_1$ is empty.
At time~5, $\tau_1$ skips since $c_2$ is now empty.
At time~6, $\tau_2$ reads the token from $c_1$ and writes $c_2$ at its deadline, time~10.
At that moment $\tau_1$ is also released, and by the write-before-read rule it reads the new token.

\subsection{Formal Semantics}
\label{sec:semantics}

We now formally define the behavior of MIMOS models as a labeled transition system and establish the determinism property stated informally above.

\myparagraph{System Model and Configuration}
A \textbf{System Model} in MIMOS consists of a finite set of tasks $\setoftasks$ and a finite set of channels $\setofchannels$, forming a directed graph where nodes are tasks, edges are channels, and each channel is typed as either a FIFO or a register. Tasks are ranged over by $\tau$ and channels by $c$.
Each task $\tau \in \setoftasks$ is defined by a tuple $(\period(\tau), \deadline(\tau), \inputchannels(\tau), \outputchannels(\tau), \taskfunction(\tau))
$,
where $\period(\tau), \deadline(\tau) \in \mathbb{N}_{>0}$ stand for the period and relative deadline\footnote{Both are positive, as naturally required for feasibility. However, the deadline can be set to zero to emulate the zero-time semantics of synchronous models supported by the current implementation of MIMOS-Tools.
}, $\inputchannels(\tau), \outputchannels(\tau) \subseteq \setofchannels$ denote the input and output channels, and $\taskfunction(\tau)$ specifies the task function.
We assume that for any two distinct tasks $\tau, \tau'$, $\inputchannels(\tau) \cap \inputchannels(\tau') = \emptyset$ and $\outputchannels(\tau) \cap \outputchannels(\tau') = \emptyset$, and that $\bigcup_{\tau \in \setoftasks}\inputchannels(\tau)$ = $\bigcup_{\tau \in \setoftasks}\outputchannels(\tau) = \setofchannels$,
i.e., each channel has exactly one reader and exactly one writer.

Let $\setofdatavalues$ denote the domain of token values, where $\setofdatavalues^*$ is the set of all finite sequences over $\setofdatavalues$ and $\varepsilon \in \setofdatavalues^*$ represents the empty sequence.
Since each task may have multiple input and output channels, we let $\outvectorofdatavalues$ range over vectors of token values, with one component per channel.
A scheduled write operation is a tuple $(\tau, d, \outvectorofdatavalues)$,
meaning that task $\tau$ writes $\outvectorofdatavalues$ to its output channels at the future time $d$.
A {\it System Configuration} is a tuple
$
  (t, \State, \setofscheduledwrites, \setofreleasedtasks)
$,
where
\begin{itemize}
    \item $t \in \mathbb{N}_{\geq 0}$ is the current time. 
    \item $\State : \setofchannels \to \setofdatavalues^*$  is the state, mapping channels to their current contents\footnote{Register channels store a single data value, viewed as a length-one sequence in $\setofdatavalues^*$.}.
    \item  $\setofscheduledwrites$ is the set of write operations scheduled to commit at future time instants.
    \item $\setofreleasedtasks \subseteq \setoftasks$ is the set of tasks released at $t$ and ready to read their inputs.
\end{itemize}
We assume that the initial configuration of a system is 
$(0, \State_0, \emptyset, R_0)$,
where time starts at $0$, the initial contents of channels are given by $\State_0$, and $R_0 = \mathcal{T}$ (all tasks released at time $0$).

\myparagraph{Operational Semantics}
To formalize the transition rules, we define the precise 
meaning of read and write operations on channels.
The read operation $(\outvectorofdatavalues, \State') = \readop(\tau, \State)$ reads the input channels of $\tau$.
For each $c \in \inputchannels(\tau)$, if $c$ is a FIFO then $\outvectorofdatavalues(c)$ is the head token of $\State(c)$ and $\State'(c)$ is $\State(c)$ with its head removed.
If $c$ is a register then $\outvectorofdatavalues(c) = \State'(c) = \State(c)$.
Channels not in $\inputchannels(\tau)$ remain unchanged.
The write operation $\State' = \writeop(\tau, \State, \outvectorofdatavalues)$ updates the output channels of $\tau$.
For each $c \in \outputchannels(\tau)$, if $c$ is a FIFO then $\outvectorofdatavalues(c)$ is appended to the end of $\State(c)$.
If $c$ is a register then $\State(c)$ is overwritten with $\outvectorofdatavalues(c)$.
Channels not in $\outputchannels(\tau)$ remain unchanged.
%
%
We use
$\released(t) \triangleq \{\tau \in \setoftasks \mid \exists k \in \mathbb{N}.\; t = k \cdot \period(\tau)\}$
to stand for the set of tasks released at time $t$, and $\setofscheduledwrites(t) \triangleq \{ (\tau, d, \outvectorofdatavalues) \mid (\tau, d, \outvectorofdatavalues) \in \setofscheduledwrites \wedge d = t \}$
for  the set of write-operations scheduled for time $t$.
Furthermore, we define the next event time
$\mathsf{NextTime}(t) \triangleq \min\bigl( \{ k \cdot \period(\tau) \mid \tau \in \setoftasks,\, k \in \mathbb{N},\, k \cdot \period(\tau) > t \} \cup \{ k \cdot \period(\tau) + \deadline(\tau) \mid \tau \in \setoftasks,\, k \in \mathbb{N},\, k \cdot \period(\tau) + \deadline(\tau) > t \} \bigr)$,
which is the earliest future time instant at which either a task is released or a deadline occurs.
A task $\tau$ can compute if it has all inputs available, denoted $\enabled(\tau, \State)$, which holds iff
$\forall c \in \inputchannels(\tau).\; \State(c) \neq \varepsilon$.

\begin{semanticsbox}[title=Transition Rules]
  \footnotesize 
  \begin{mathpar}
    \inferrule*[lab=\textsc{Read}]{
      \setofscheduledwrites(t) = \emptyset \\ 
      \tau \in \setofreleasedtasks \\
      \enabled(\tau, \State) \\
      (\outvectorofdatavalues, \State') = \readop(\tau, \State)
    }{
      (t, \State, \setofscheduledwrites, \setofreleasedtasks) \longrightarrow
      (t, \State', \setofscheduledwrites \cup \{(\tau,\, t + \deadline(\tau),\, \taskfunction(\tau)(\outvectorofdatavalues))\}, \setofreleasedtasks \setminus \{\tau\})
    }
    
    \inferrule*[lab=\textsc{Write}]{
      (\tau, t, \outvectorofdatavalues) \in \setofscheduledwrites(t) \\
      \State' = \writeop(\tau, \State, \outvectorofdatavalues)
    }{
      (t, \State, \setofscheduledwrites,\setofreleasedtasks) \longrightarrow
      (t, \State', \setofscheduledwrites \setminus \{(\tau, t, \outvectorofdatavalues)\}, \setofreleasedtasks)
    }
    
    \inferrule*[lab=\textsc{Next}]{
      \setofscheduledwrites(t) = \emptyset \\
      \forall \tau \in \setofreleasedtasks.\; \lnot\enabled(\tau, \State) \\
      t' = \mathsf{NextTime}(t)
    }{
      (t, \State, \setofscheduledwrites, \setofreleasedtasks) \longrightarrow
      (t', \State, \setofscheduledwrites, \released(t'))
    }
  \end{mathpar}
\end{semanticsbox}

The operational semantics is defined by three transition rules:
(1)
\textsc{Read} fires at time $t$ when no write operations are scheduled at $t$, 
i.e., \ $\setofscheduledwrites(t) = \emptyset$, and there exists a released task $\tau \in \setofreleasedtasks$ enabled in the current state $\State$.
The inputs of $\tau$ will be read, assigned to $\outvectorofdatavalues$, and will be used to compute the outputs $\taskfunction(\tau)(\outvectorofdatavalues)$ to write at the deadline of $\tau$, i.e., $t + \deadline(\tau)$.
The precondition $\setofscheduledwrites(t) = \emptyset$ enforces the write-before-read rule, so all writes scheduled at $t$ complete before any read at that instant.
(2)
\textsc{Write} commits the write operations scheduled for time $t$ one by one.
(3)
\textsc{Next} advances time when no write operations are pending at $t$ and no released tasks in $\setofreleasedtasks$ are enabled.
Time advances to $t' = \mathsf{NextTime}(t)$, and $\setofreleasedtasks$ is reset to the set of tasks released at $t'$.

\begin{theorem}[Determinism]
Let $\gamma$ be a reachable configuration at time $t$.
Any two maximal sequences of \textsc{Write} and \textsc{Read}
transitions starting from $\gamma$, before applying \textsc{Next},
reach the same configuration.
\end{theorem}

\begin{proof}
By induction on the length of transitions.
\end{proof}


\section{Tool Architecture}
\label{sec:layer}
 
MIMOS-Tools decouples logical design from physical deployment through a layered architecture inspired by AUTOSAR ~\cite{autosar2024standard}.
Unlike AUTOSAR, which primarily separates application software from runtime and basic software services, MIMOS-Tools organizes its layers around the main engineering concerns in embedded system development: functional correctness, timing correctness, and platform deployment.

\begin{figure}[!htbp]
    \centering
    \includesvg[width=1.0\linewidth]{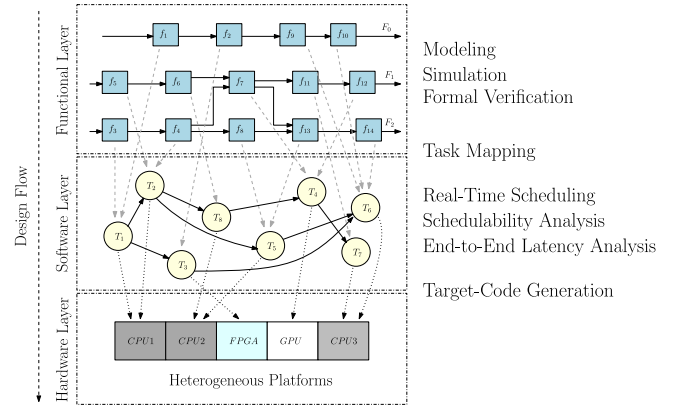}
    \caption{The layered architecture of MIMOS-Tools and the 
    key features of its individual tools. Left-to-right arrows 
    indicate data flows through communication channels, while 
    dashed top-to-bottom arrows represent task mapping 
    (currently specified manually by users).}
    \Description{Layered architecture of MIMOS-Tools}
    \label{fig:layers}
\end{figure}

As illustrated in \autoref{fig:layers}, the top \textit{functional layer} represents systems as MIMOS networks. At this layer, designers specify task behavior and channel connections without considering schedulability or resource allocation. MIMOS-Tools supports graphical modeling, visual simulation, static checks for deadlock freedom and FIFO boundedness, and functional verification. Since this layer is based on the deterministic MIMOS semantics, functional properties validated here are preserved throughout the design flow. The middle \textit{software layer} refines verified MIMOS components into real-time tasks for scheduling and timing analysis. 
Thanks to non-blocking communication, tasks can be analyzed as independent real-time tasks, without modeling blocking delays.
This layer supports schedulability analysis under a given scheduling policy and platform configuration, as well as end-to-end latency analysis from source tasks to observable outputs. The bottom \textit{hardware layer} deploys validated models onto target platforms. The code generation module transforms MIMOS models into executable programs, while preserving the deterministic communication semantics established at the functional layer and respecting the timing constraints verified at the software layer.

This separation of concerns allows designers to reason about functional correctness independently of scheduling and deployment decisions. 
Consequently, as long as all tasks remain schedulable, changes to scheduling policies or platform configurations do not invalidate the functional behavior already validated at the model level, enabling efficient design-space exploration for embedded real-time systems.
\section{Tool Features and Case Studies}
\label{sec:features}

MIMOS-Tools is composed of a set of individual tools, supporting the different stages of the development process.
We illustrate the tool features using two case studies.
The first is a simplified UAV altitude controller with four tasks for command processing, control computation, sensing, and UAV dynamics.
The second is a perception-action pipeline for autonomous driving, adapted from the automotive object-detection application  in ~\cite{kang2022dnnsam}, with six tasks for camera acquisition, preprocessing, region-of-interest detection, full-scene detection, detection fusion, and engine control.
Its two detection tasks form a dual-path structure: a high-priority region-of-interest path for safety-critical detection and a lower-rate full-scene path for contextual information.

\myparagraph{Modeling and Simulation}
MIMOS-Tools provides a graphical modeling environment for constructing MIMOS models through a GUI, as shown in \autoref{fig:model_and_properties}.
Users build a system by creating tasks and channels in a drag-and-drop manner.
Each task can be configured through the properties panel with timing parameters, associated input/output channels, and channel types, either FIFOs or registers.

Although MIMOS was designed to overcome the lack of modularity and updatability of synchronous design, the current implementation of MIMOS-Tools supports Simulink models by setting the deadlines of tasks to zero to emulate the zero-time semantics of synchronous design when the models are checked (automatically by MIMOS-Tools) and found to contain no causality errors.

\begin{figure}[!htbp]
  \centering
  \includegraphics[width=0.85\linewidth]{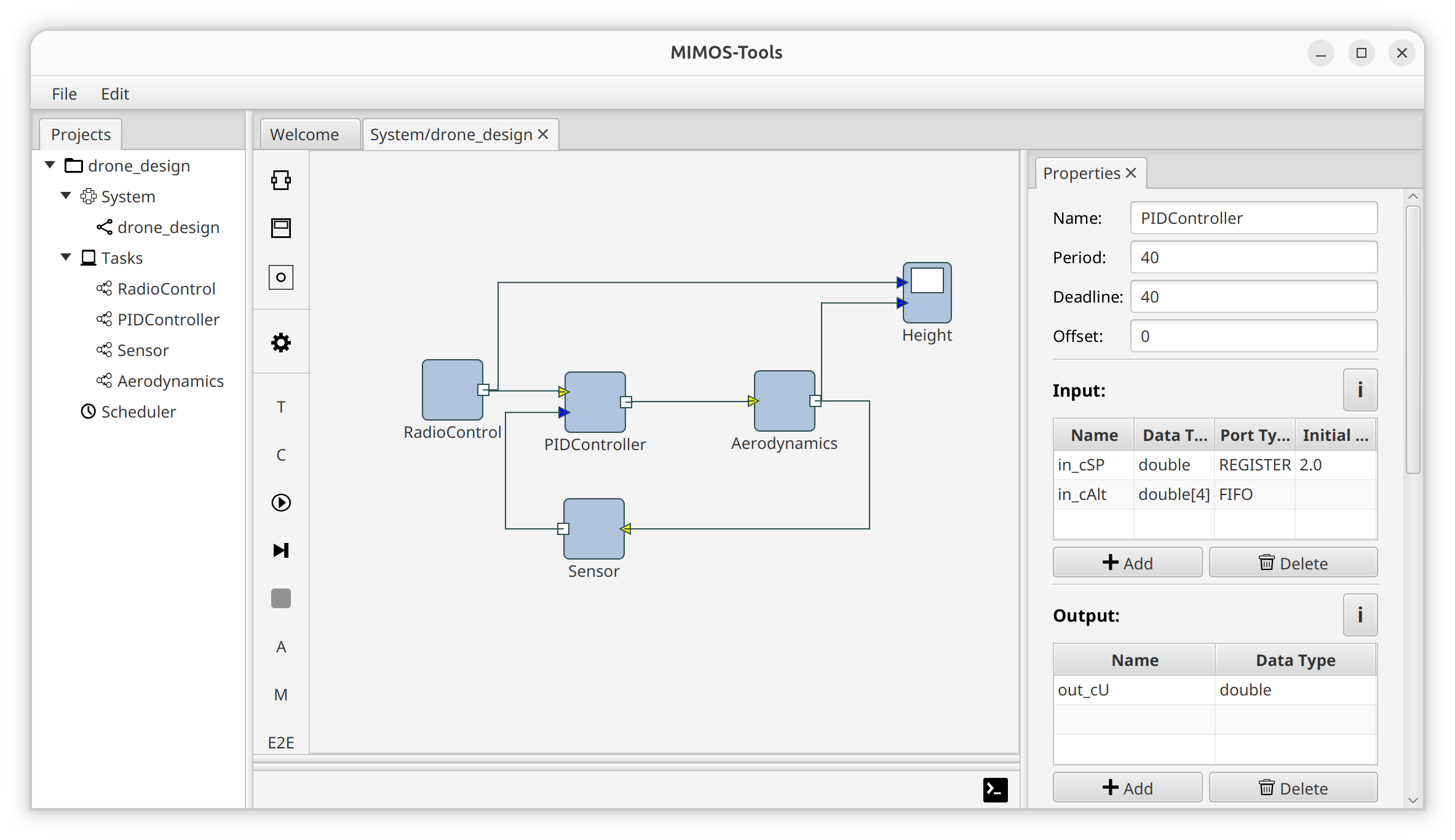}
  \caption{Graphical modeling interface with properties panel of the UAV
  altitude controller, consisting of four tasks connected by FIFO and register channels.}
  \Description{Screenshot of the graphical modeling interface showing the UAV altitude controller composed of four tasks connected by FIFO and register channels, with a properties panel visible.}
  \label{fig:model_and_properties}
\end{figure}

\begin{figure*}[!tb]
  \centering
  \begin{minipage}[t]{0.32\textwidth}
    \centering
    \includegraphics[width=\linewidth]{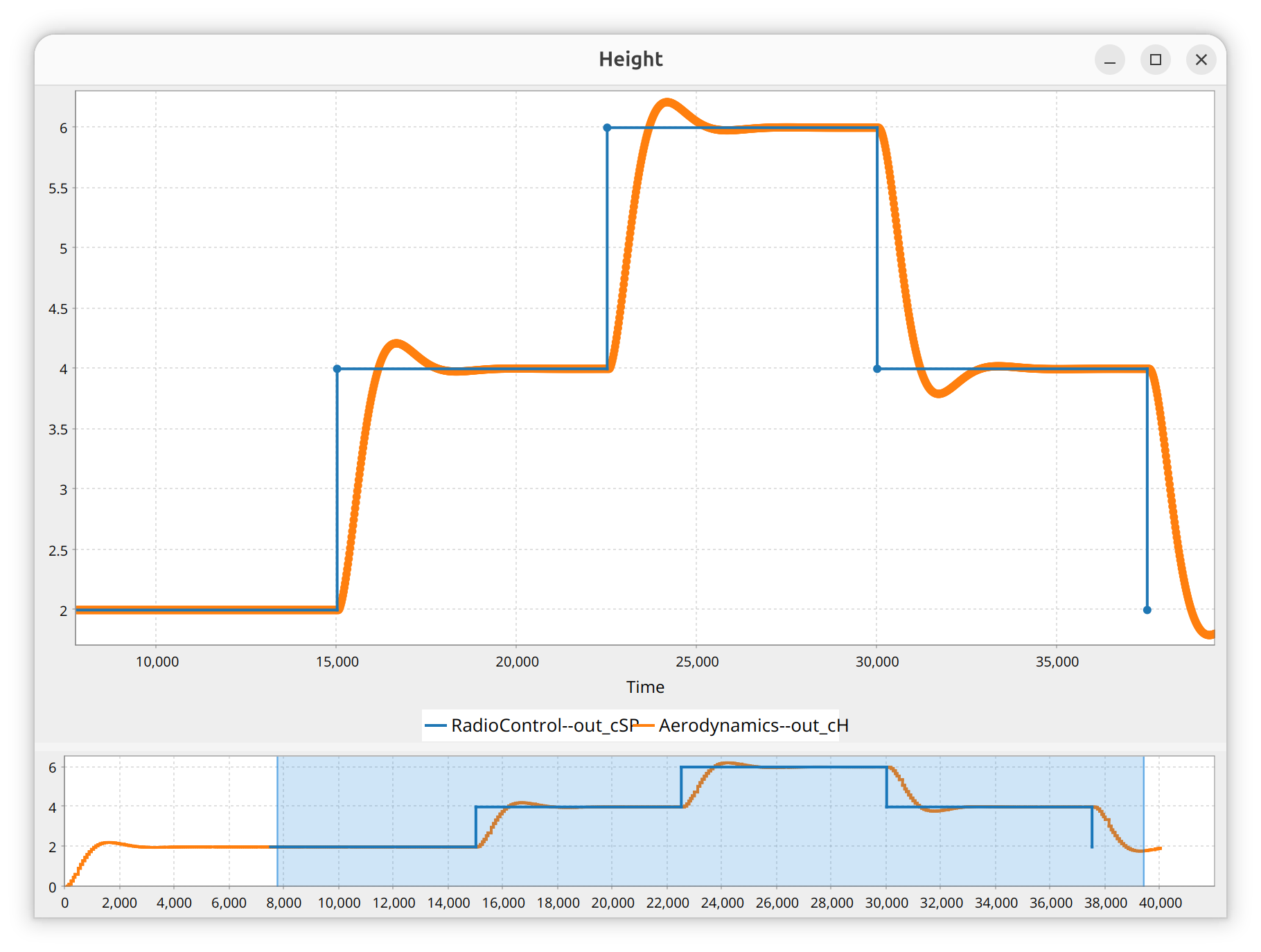}
    \vspace{-4pt}
    \centerline{(a) Simulation result}
  \end{minipage}
  \hfill
  \begin{minipage}[t]{0.32\textwidth}
    \centering
    \includegraphics[width=\linewidth]{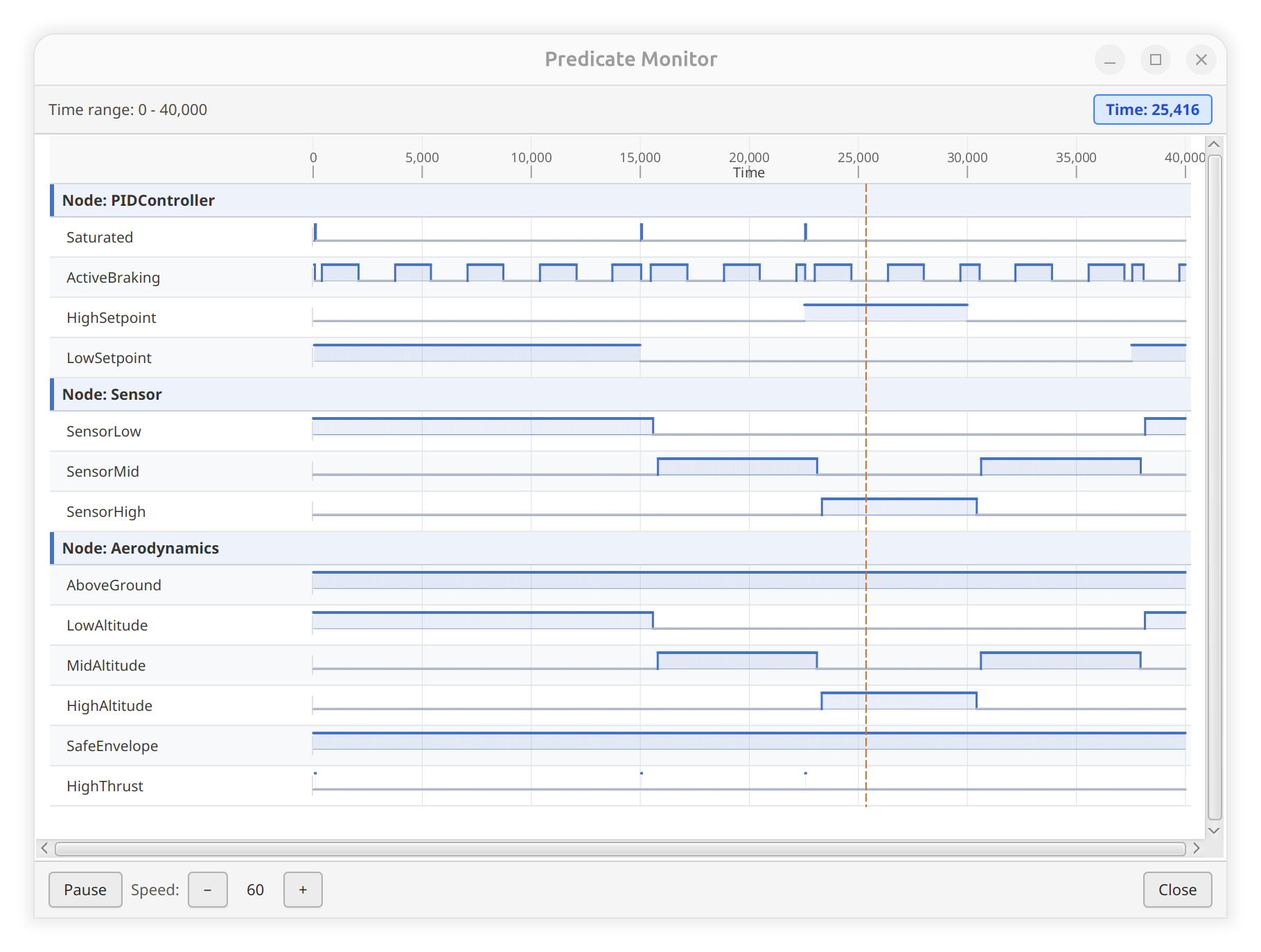}
    \vspace{-4pt}
    \centerline{(b) Predicate monitoring}
  \end{minipage}
  \hfill
  \begin{minipage}[t]{0.32\textwidth}
    \centering
    \includegraphics[width=\linewidth]{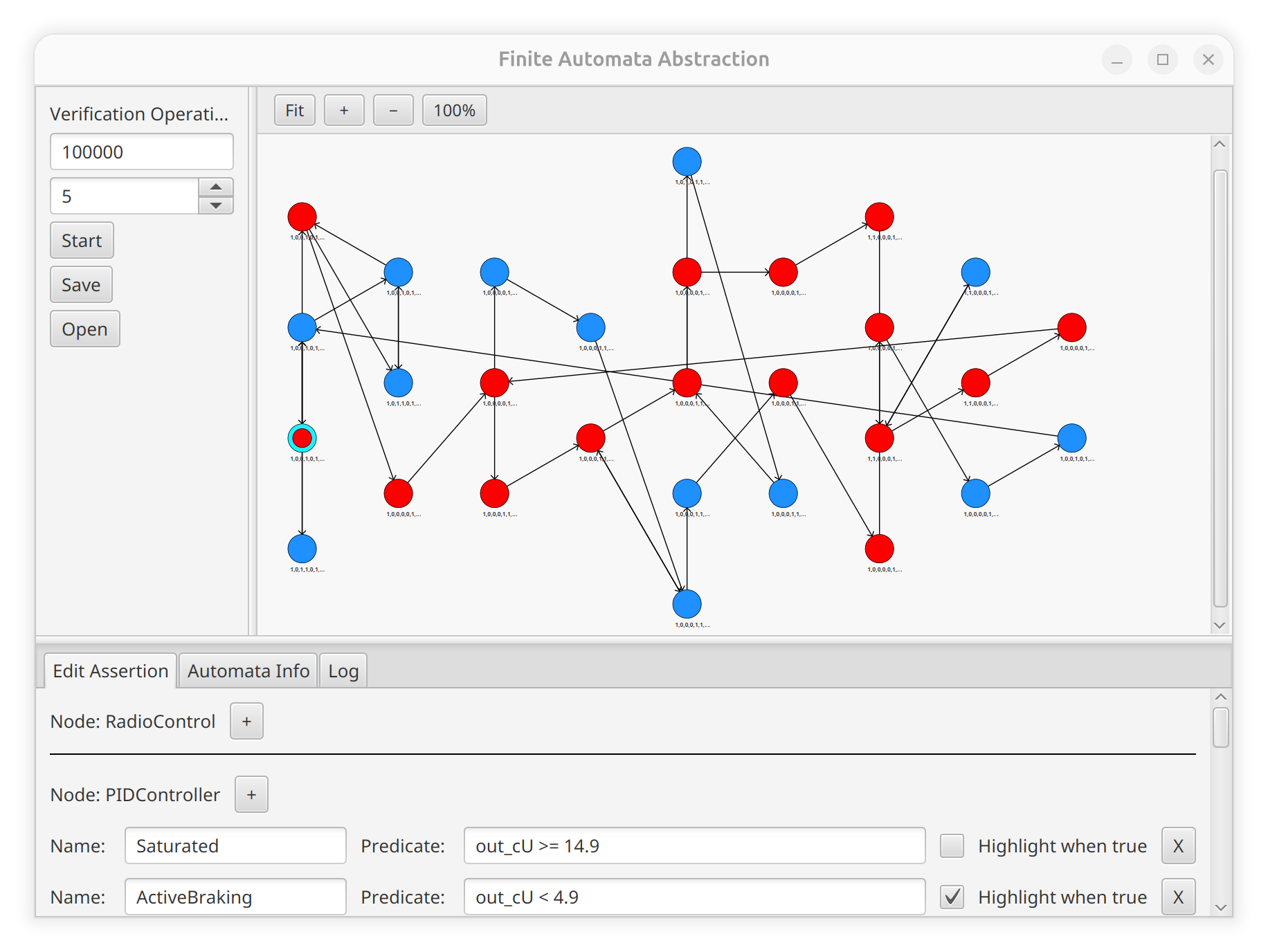}
    \vspace{-4pt}
    \centerline{(c) Finite-automaton abstraction}
  \end{minipage}

   \caption{Simulation, monitoring, and verification of the UAV altitude controller.
  \textbf{(a)}~Altitude tracking across four setpoints.
  \textbf{(b)}~Truth values of the predicates over time.
  \textbf{(c)}~Finite-automaton abstraction built from the predicates.}
  \Description{Three horizontally arranged screenshots from MIMOS-Tools. The first shows simulated altitude tracking for the UAV controller across four setpoints. The second shows predicate-monitoring results over time. The third shows a finite automaton abstraction used for verification.}
  \label{fig:sim_mon_verify}
\end{figure*}

\begin{figure*}[!t]
  \centering
  \includegraphics[height=0.115\textheight]{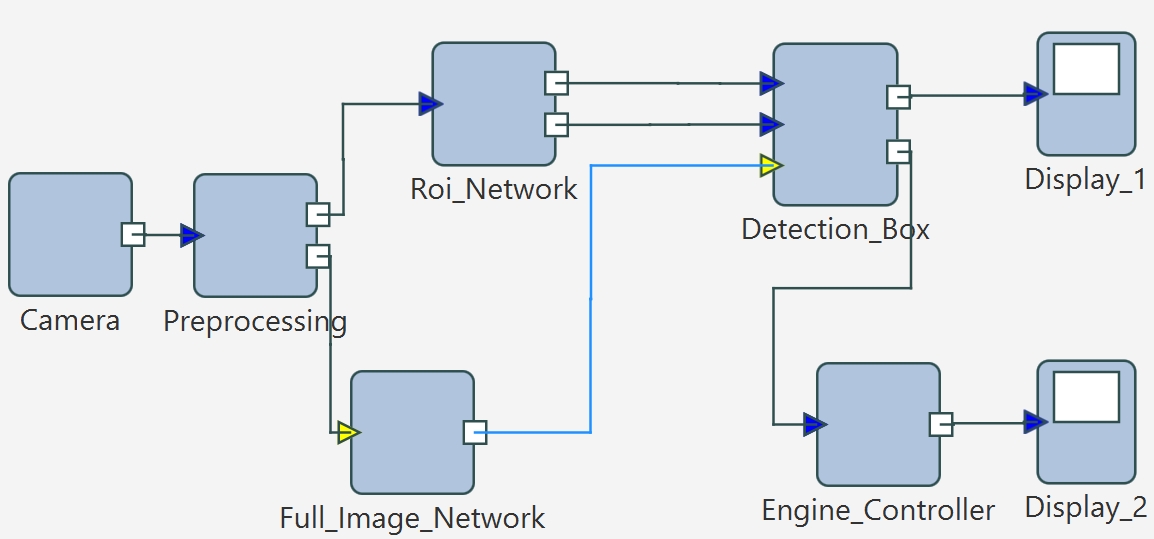}
  \hfill
  \includegraphics[height=0.115\textheight]{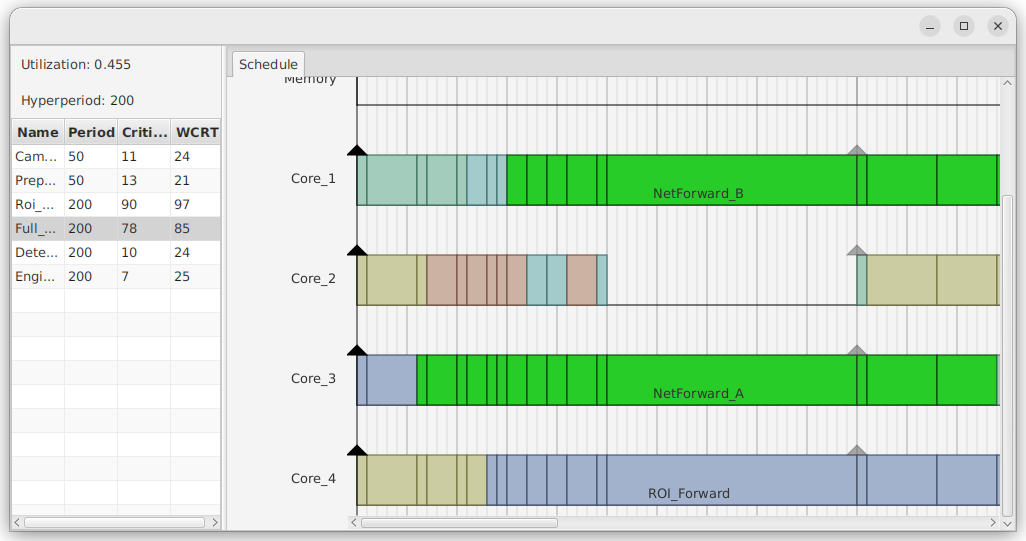}
  \hfill
  \includegraphics[height=0.115\textheight]{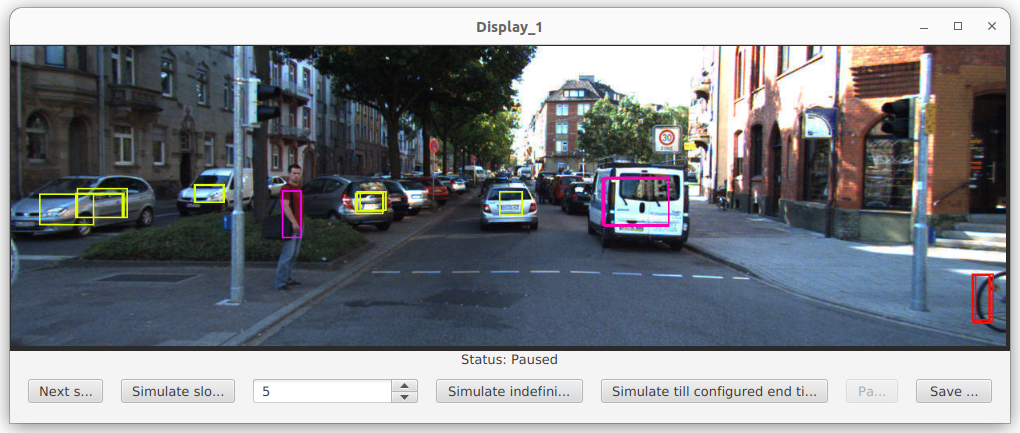}

  \vspace{-3pt}

  \begin{minipage}[t]{0.30\textwidth}\centering (a) MIMOS model\end{minipage}
  \hfill
  \begin{minipage}[t]{0.30\textwidth}\centering (b) Real-time scheduling and analysis\end{minipage}
  \hfill
  \begin{minipage}[t]{0.30\textwidth}\centering (c) Object detection in simulation\end{minipage}

  \caption{Real-time perception-action pipeline for autonomous driving:
  (a) MIMOS model, (b) Response-time analysis and Gantt chart for Global EDF on a four-core platform, and (c) Simulation of traffic scenarios and object detection.}
  \Description{Three screenshots showing the autonomous driving perception-action case study: the MIMOS model, the scheduling analysis result, and the object-detection simulation.}
  \label{fig:av_case_study}
\end{figure*}

\textit{Hierarchical modeling:}
To support large designs, MIMOS-Tools provides two approaches to hierarchical modeling.
Groups of tasks and channels can be collapsed into \textit{subsystems} to reduce visual complexity without changing the model semantics.
Individual tasks can also be defined as \textit{composite nodes}, whose functionality is represented as a DAG of sub-functions.
Composite nodes expose internal parallelism for scheduling and code generation, while still behaving as single MIMOS tasks following the read--compute--write pattern described in \autoref{sec:model}.

\textit{Mixed-language and LLM-assisted development:}
Task functions can be written in C++ or Java and mixed in one model for rapid prototyping.
MIMOS-Tools also supports \textit{natural-language specification} of task functions: users describe what a task should compute, and use an LLM to generate the corresponding implementation.

\textit{Simulation and visualization:}
MIMOS-Tools simulates models by executing the deterministic semantics step by step.
Users specify the simulation duration and visualize selected task outputs as time-series signals.
They can also define \textit{predicates} over model configurations and monitor their truth values during simulation.
\autoref{fig:sim_mon_verify}(a) and (b) show these results for the UAV controller.
\autoref{fig:av_case_study}(c)  shows the simulation of the perception-action pipeline with traffic scenarios and object detection.

\myparagraph{Real-Time Scheduling and Analysis}
The nodes of a MIMOS model represent real-time tasks characterized by a period, deadline, and WCET.
Composite nodes can be represented as DAG tasks to expose internal parallelism among sub-functions.
Since communication never blocks, existing schedulability analyses for independent tasks apply directly.
Users configure the analysis by selecting a scheduling policy and specifying the target platform configuration.
MIMOS-Tools checks the schedulability and computes the worst-case response times of the tasks.
Currently more than 20 multicore scheduling policies are supported. 
\autoref{fig:av_case_study}(b) shows the analysis for the perception-action pipeline on a four-core platform under Global EDF. Beyond response times, the tool also computes \textit{worst-case end-to-end latency} for user-specified source--sink paths through the MIMOS network.
Intuitively, the end-to-end latency of a path measures the maximum delay from a data change at a source task's input to its observable effect at a sink task's output, after propagation through all intermediate tasks and channels.

\myparagraph{Formal Verification}
Simulation generates only bounded execution trajectories but cannot prove properties over the full reachable state space.
MIMOS-Tools therefore provides verification support through static analysis and state-space exploration.
The static analysis checks properties including \textit{deadlock freedom} and \textit{FIFO boundedness} using only the model structure, timing parameters, and initial tokens.
When both checks succeed, the tool reports task utilization and peak FIFO sizes, which are used as buffer bounds during code generation.
For functional verification, users define local predicates over task inputs and outputs and specify global properties as LTL formulas over these predicates. The current implementation supports the verification of invariant, reachability, and leads-to properties. Given a MIMOS model, the verifier explores its reachable state space and constructs a finite-automaton abstraction on the fly, with each automaton state corresponding to a truth-value combination of the selected predicates. \autoref{fig:sim_mon_verify}(c) shows the automaton generated for the UAV controller from 13 local predicates, with highlighted states indicating those that satisfy a selected predicate. Global temporal properties can then be verified using standard model-checking algorithms.

\myparagraph{Code Generation}
\label{subsec:code_generation}
After passing the functional and timing checks, MIMOS-Tools generates executable code for multicore deployment.
The generated program combines user-defined task functions with a coordination layer implementing the MIMOS semantics and a real-time scheduler realizing the analyzed scheduling policy.
\section{Conclusion}
MIMOS-Tools is an integrated toolchain supporting the entire development lifecycle of embedded real-time systems, from design and implementation to dynamic updates after deployment. 
In this paper, we have presented the concise semantics of the MIMOS model together with the main features of the toolchain, as well as its modeling and analysis capabilities demonstrated through examples from avionics and automotive applications.

Compared with traditional approaches, such as  SCADE and Simulink, which are grounded in the synchronous paradigm and provide deterministic design, MIMOS-Tools also ensures deterministic design while enabling modular and incremental system development, facilitating efficient and reliable software updates throughout long-term system evolution.

\section*{Declarations}
\textbf{Funding Statement.}
This work was supported in part by the European Research Council (ERC) and the Knut and Alice Wallenberg Foundation (KAW) through projects ERC CUSTOMER under Grant 834166 and KAW UPDATE under Grant 20190134. 

\noindent\textbf{Data-Availability Statement.}
The artifact accompanying this paper is archived at \url{https://doi.org/10.5281/zenodo.21637714}.

\bibliographystyle{ACM-Reference-Format}
\bibliography{references}

\end{document}